\documentclass[letterpaper, 10 pt, conference]{ieeeconf}

\IEEEoverridecommandlockouts

\usepackage{amssymb,latexsym,amsfonts,amsmath}
\usepackage{graphicx}
\usepackage{paralist}
\usepackage{cite}
\usepackage{eso-pic}
\usepackage{epsfig}
\usepackage{mathtools}
\usepackage{tikz}
\usepackage{epstopdf}
\usepackage{ltl}
\usetikzlibrary{calc,shapes,arrows}
\usepackage[normalem]{ulem}
\usepackage{comment}

\newtheorem{theorem}{Theorem}
\newtheorem{lemma}{Lemma}

\newtheorem{definition}{Definition}

\newtheorem{remark}{Remark}

\newcommand{\R}{{\mathbb{R}}}

\newcommand{\N}{{\mathbb{N}}}

\newcommand{\eg}{{\it e.g.}}

\newcommand\norm[1]{\left\lVert#1\right\rVert}

\newcommand{\Let}{:=}

\newcommand{\Z}{\mathbb{Z}}

\newcommand{\intcc}[1]{\ensuremath{{\left[#1\right]}}}

\renewcommand{\emptyset}{{\varnothing}}

\title{
	{From dissipativity theory to compositional synthesis of symbolic models
		\thanks{This work was supported in part by the TUM International Graduate School of Science and Engineering (IGSSE).}}
}
\author{Abdalla Swikir, Antoine Girard, and Majid Zamani
	\thanks{A. Swikir and M. Zamani are with the Department of Electrical and Computer Engineering, Technical University of Munich, D-80290 Munich, Germany. A. Girard is with the Laboratoire des Signaux et Syst\'emes (L2S) - CNRS, 91192 Gif sur Yvette, France.
		Email: {\tt\small \{abdalla.swikir,zamani\}@tum.de, Antoine.Girard@l2s.centralesupelec.fr}.
}}

\begin{document}

	\maketitle
	%	\thispagestyle{empty}
	%	\pagestyle{empty}
	%	
	
	%%%%%%%%%%%%%%%%%%%%%%%%%%%%%%%%%%%%%%%%%%%%%%%%%%%%%%%%%%%%%%%%%%%%%%%%%%%%%%%%
	\begin{abstract}
		
		In this work, we introduce a compositional framework for the construction of finite abstractions (a.k.a. symbolic models) of
		interconnected discrete-time control systems. The compositional scheme is based on the joint dissipativity-type
		properties of discrete-time control subsystems and their finite abstractions. In the first part of the paper, we use a notion of so-called storage function as a relation between each subsystem and its finite abstraction to construct compositionally a notion of so-called simulation function as a relation between interconnected finite abstractions and that of control systems. The derived simulation function is used to quantify the error between the
		output behavior of the overall interconnected concrete system and that of its finite abstraction. 
		In the second part of the paper, we propose a technique to construct finite abstractions together with their corresponding storage functions for a class of discrete-time control systems under some incremental passivity property.
		We show that if a discrete-time control system is so-called incrementally passivable, then one can construct its finite abstraction by
		a suitable quantization of the input and state sets together with the corresponding storage function. Finally, the proposed results are illustrated by constructing a finite abstraction of a network of linear discrete-time control systems and its corresponding simulation function in a compositional way. The compositional conditions in this example do not impose any restriction on the gains or the number of the subsystems which, in particular, elucidates the effectiveness of dissipativity-type compositional reasoning for networks of systems.
	\end{abstract}

	%%%%%%%%%%%%%%%%%%%%%%%%%%%%%%%%%%%%%%%%%%%%%%%%%%%%%%%%%%%%%%%%%%%%%%%%%%%%%%%%
	\section{Introduction}

	In the recent years, symbolic models were introduced as a method to reduce the complexity of controller synthesis in particular for enforcing complex logical properties. Symbolic models (a.k.a. finite abstractions) are abstract descriptions of the continuous-space control systems in which each discrete state corresponds to a collection of continuous states of the original system. Since symbolic models are finite, algorithmic approaches from computer science are applicable to synthesize controllers enforcing some complex properties including those expressed as linear temporal logic formulae.     
	
	Large-scale interconnected control systems, \eg, biological networks, power networks, and manufacturing systems, are intrinsically difficult to analyze and control and it is very challenging to design a controller to achieve some complex logical specifications over those interconnected systems. An appropriate technique to overcome this challenge is to first treat every subsystem individually and build an abstraction that approximates the behaviors of the corresponding concrete subsystem. Thereafter, one can establish a compositional framework to construct abstractions of the network of control subsystems and use them as a replacement in the controller design process. Recently, there have been several results on the compositional construction of finite abstractions of networks of linear and nonlinear control systems in \cite{Tazaki2008,7403879}. Recent work on the compositional construction of infinite abstractions of interconnected nonlinear control systems can be found in \cite{Rungger,7496809}. The results in \cite{Tazaki2008,7403879,Rungger,7496809} use the small-gain type conditions to facilitate the compositional construction of (in)finite abstractions. However, those small-gain type conditions depend essentially on the size of the network graph and can be violated as the number of subsystems increases \cite{Das2004149}. The recent results in \cite{7857702} propose a compositional framework for the construction of infinite abstractions of networks of contiunous-time control systems using dissipativity theory \cite{murat}. The
proposed compositionality conditions in \cite{7857702} can enjoy specific interconnection
topologies and provide scale-free compositional
abstractions for large-scale control systems.
	
	In this work, we introduce a compositional approach for the construction of \emph{finite} abstractions of interconnected discrete-time control systems using techniques from dissipativity theory \cite{murat}. First, we introduce a notion of so-called storage function inspired by the one introduced in \cite{7857702} and use it to quantify the joint dissipativity-type properties of discrete-time control subsystems and their finite abstractions. Given storage functions between subsystems and their finite abstractions, we drive compositional conditions under which one can construct a so-called simulation function, similar to the one introduced in \cite{Girard2009566}, as a relation between the interconnected abstractions and the concrete network of control subsystems. The existence of such a simulation function ensures that the output behavior of the concrete system is quantitatively approximated by the corresponding one of its finite abstraction. In addition, we provide a procedure for the construction of finite abstractions together with their corresponding storage functions for a class of discrete-time control systems satisfying some incremental passivity property. Finally, we demonstrate the effectiveness of our results on an interconnected discrete-time linear control system in which the compositionality condition is always satisfied independently of the number of subsystems.
	
%	The rest of the paper is organized as follows. In Section \ref{1:II}, we introduce some mathematical preliminaries and the definition of discrete-time control (sub)systems, and also the definitions of so-called storage and simulation functions. In Section \ref{1:III}, we define the interconnected discrete-time control systems and show how to compose simulation functions from storage functions of subsystems. In Section \ref{1:IV}, we propose a scheme to construct finite abstractions together with their corresponding storage functions for a class of so-called incrementally passivable discrete-time control systems. In Section \ref{1:V}, we demonstrate the effectiveness of our results on an interconnected discrete-time linear control system in which the compositionality condition is always satisfied independently of the number of subsystems.		
	
	\section{Notation and Preliminaries}\label{1:II}
	\subsection{Notation}
	We denote by $\R$, $\Z$, and $\N$ the set of real numbers, integers, and non-negative integers,  respectively.
	These symbols are annotated with subscripts to restrict them in
	the obvious way, \eg, $\R_{>0}$ denotes the positive real numbers. We denote the closed, open and half-open intervals in $\R$ by $[a,b]$,
	$(a,b)$, $[a,b)$, and $(a,b]$, respectively. For $a,b\in\N$ and $a\le b$, we
	use $[a;b]$, $(a;b)$, $[a;b)$, and $(a;b]$ to
	denote the corresponding intervals in $\N$.
	Given $N\in\N_{\ge1}$, vectors $\nu_i\in\R^{n_i}$, $n_i\in\N_{\ge1}$, and $i\in[1;N]$, we
	use $\nu=[\nu_1;\ldots;\nu_N]$ to denote the vector in $\R^N$ with
	$N=\sum_i n_i$ consisting of the concatenation of vectors~$\nu_i$. Note that given any  $\nu\in\R^{n}$, $\nu \ge 0$ iff $\nu_i \ge 0$ for any $i \in [1;n]$. We denote by $\mathsf{diag}(M_1,\ldots,M_N)$ the block diagonal matrix with diagonal matrix entries $M_1,\ldots,M_N$. We denote the identity and zero matrices in $\R^{n\times n}$ by $I_n$ and $0_{n}$, respectively. Given a function $f: \N \rightarrow \R^n $, the supremum of $ f $ is denoted by $ \|f\|_\infty $; we recall that $ \|f\|_\infty$ := sup$\{\|f(k)\|,k\geq0\}$, where $\norm{\cdot}$ denote the infinity norm. Given a function $f:\R^n\to \R^m$ and $\overline x\in\R^m$, we use $f\equiv \overline x$ to denote that $f(x)=\overline x$ for all $x\in\R^n$. If
	$\overline x$ is the zero vector, we simply write $f\equiv 0$. Given a set $A$ and matrix $P$ of appropriate dimension, $PA\Let\{Pa | a\in A\}$. The identity map on a set $A$ in denoted by $1_A$. The closed ball centered at $x\in{\mathbb{R}}^{n}$ with radius $\varepsilon$ is defined by \mbox{$\mathcal{B}_{\varepsilon}(x)=\{y\in{\mathbb{R}}^{n}\,|\,\Vert x-y\Vert\leq\varepsilon\}$}. For any set \mbox{$A\subseteq\R^n$} of the form of finite union of boxes, \eg, $A=\bigcup_{j=1}^MA_j$ for some $M\in\N$, where $A_j=\prod_{i=1}^n [c_i^j,d_i^j]\subseteq \R^n$ with $c^j_i<d^j_i$, and positive constant $\eta\leq\emph{span}(A)$, where $\emph{span}(A)=\min_{j=1,\ldots,M}\eta_{A_j}$ and \mbox{$\eta_{A_j}=\min\{|d_1^j-c_1^j|,\ldots,|d_n^j-c_n^j|\}$}, define \mbox{$[A]_{\eta}=\{a\in A\,\,|\,\,a_{i}=k_{i}\eta,k_{i}\in\mathbb{Z},i=1,\ldots,n\}$}. The set $[A]_{\eta}$ will be used as a finite approximation of the set $A$ with precision $\eta$. Note that $[A]_{\eta}\neq\varnothing$ for any $\eta\leq\emph{span}(A)$. 
	%	Geometrically, for any $\eta\in{\mathbb{R}^+}$ and $\lambda\geq\eta$ the collection of sets \mbox{$\{\mathcal{B}_{\lambda}(p)\}_{p\in[A]_{\eta}}$} is a covering of $A$, i.e. \mbox{$A\subseteq\bigcup_{p\in[A]_{\eta}}\mathcal{B}_{\lambda}(p)$}. 
	We denote by $|\cdot|$ the cardinality of a given set and by $\emptyset$ the empty set. 
	We use notations $\mathcal{K}$, $\mathcal{K}_\infty$, and $\mathcal{KL}$
	to denote the different classes of comparison functions, as follows:
	$\mathcal{K}=\{\alpha:\mathbb{R}_{\geq 0} \rightarrow \mathbb{R}_{\geq 0} |$ $ \alpha$ is continuous, strictly increasing, and $\alpha(0)=0\}$; $\mathcal{K}_\infty=\{\alpha \in \mathcal{K} |$ $ \lim\limits_{r \rightarrow \infty} \alpha(r)=\infty\}$; a function $\beta:\mathbb{R}_{\geq 0} \times\mathbb{N}\rightarrow \mathbb{R}_{\geq 0}$ is a $\mathcal{KL}$ function if, for each fixed $k \geq 0$, the function $\beta(\cdot,k)$ is a $\mathcal{K}$ function, and for each fixed $ r\neq0 $ the function $\beta(r,\cdot)$ is decreasing and $\beta(r,k) \rightarrow 0 $ as $ k\rightarrow\infty$. 
	
	\subsection{Discrete-time control systems} 
	In this paper we study discrete-time control systems of the following form.
	\begin{definition}\label{def:sys1}
		A discrete-time control system $\Sigma$ is defined by the tuple
		\begin{IEEEeqnarray}{c}\label{eq:1}
			\begin{IEEEeqnarraybox}[\relax][c]{rCl}
				\Sigma&=&(\mathbb X,\mathbb U,\mathbb W,\mathcal{U},\mathcal{W},f,\mathbb Y_{1},\mathbb Y_{2},h_1,h_2),
			\end{IEEEeqnarraybox}
		\end{IEEEeqnarray}
		where $\mathbb X, \mathbb U, \mathbb W, \mathbb Y_{1},$ and $\mathbb Y_{2}$ are the state set, external input set, internal input set, external output set, and internal output set, respectively, and are assumed to be subsets of normed vector spaces with appropriate finite dimensions. Sets $\mathcal{U}$ and $\mathcal{W}$ denote the set
		of all bounded input functions $\nu:\N\rightarrow \mathbb U$ and $\omega:\N\rightarrow \mathbb W$, respectively. The set-valued map $f: \mathbb X\times \mathbb U \times \mathbb W\rightrightarrows \mathbb X $ is called the transition function \cite{Rock0000}, $h_1:\mathbb X \rightarrow \mathbb Y_1$  is the external output map, and $h_2:\mathbb X\rightarrow \mathbb Y_2$ is the internal output map.
		The discrete-time conrol system $\Sigma $ is described by difference inclusions of the form
		\begin{IEEEeqnarray}{c}\label{eq:2}
		\Sigma:\left\{
			\begin{IEEEeqnarraybox}[\relax][c]{rCl}
				\mathbf{x}(k+1)&\in&f(\mathbf{x}(k),\nu(k),\omega(k))\\
				\mathbf{y}_1(k)&=&h_1(\mathbf{x}(k))\\
				\mathbf{y}_2(k)&=&h_2(\mathbf{x}(k)),
			\end{IEEEeqnarraybox}\right.
		\end{IEEEeqnarray}
		where $\mathbf{x}:\mathbb{N}\rightarrow \mathbb X $, $\mathbf{y_1}:\mathbb{N}\rightarrow \mathbb Y_1$, $\mathbf{y_2}:\mathbb{N}\rightarrow \mathbb Y_2$, $\nu\in\mathcal{U}$, and $\omega\in\mathcal{W}$ are the state signal, external output signal, internal output signal, external input signal, and internal input signal, respectively. 
		
		System $\Sigma=(\mathbb X,\mathbb U,\mathbb W,\mathcal{U},\mathcal{W},f,\mathbb Y_{1},\mathbb Y_{2},h_1,h_2)$ is called deterministic if $|f(x,u,w)|\leq1$ $ \forall x\in \mathbb X, \forall u\in \mathbb U, \forall w \in \mathbb W$, and non-deterministic otherwise. System $\Sigma$ is called blocking if $\exists x\in \mathbb X, \forall u\in \mathbb U, \forall w \in \mathbb W $ such that $|f(x,u,w)|=0$ and non-blocking if $|f(x,u,w)|\neq 0$ $ \forall x\in \mathbb X, \exists u\in \mathbb U, \exists w \in \mathbb W$. In this paper, we only deal with non-blocking systems. System $\Sigma$ is called finite if $\mathbb X,\mathbb U,\mathbb W$ are finite sets and infinite otherwise.
	\end{definition}

	\begin{remark}
		If $\Sigma$ does not have internal inputs and outputs, Definition \ref{def:sys1} reduces to the tuple $\Sigma=(\mathbb X,\mathbb U,\mathcal{U},f,\mathbb Y,h)$ and the set-valued map $f$ becomes $f:\mathbb X\times\mathbb U\rightrightarrows\mathbb X$. Correspondingly, \eqref{eq:2} reduces to:
		\begin{IEEEeqnarray}{c}\label{eq:3}
		\Sigma:\left\{
			\begin{IEEEeqnarraybox}[\relax][c]{rCl}
				\mathbf{x}(k+1)&\in&f(\mathbf{x}(k),\nu(k))\\
				\mathbf{y}(k)&=&h(\mathbf{x}(k)).
			\end{IEEEeqnarraybox}\right.
		\end{IEEEeqnarray}
		
	\end{remark}	
	
	\subsection{Storage and Simulation functions}
	First, we define a notion of so-called storage function, inspired by Definition 3.1 in \cite{7857702}, which quantifies the error between systems $\Sigma$ and $\hat{\Sigma}$ both with internal and external inputs and outputs.	
	\begin{definition}\label{def:DS}
		Consider systems
		$$\Sigma=(\mathbb X,\mathbb U,\mathbb W,\mathcal{U},\mathcal{W},f,\mathbb Y_{1},\mathbb Y_{2},h_1,h_2),$$
		and $$\hat{\Sigma}=(\mathbb{\hat{X}},\mathbb{\hat{U}},\mathbb{\hat{W}},\hat{\mathcal{U}},\hat{\mathcal{W}},\hat{f},\hat{\mathbb{Y}}_1,\hat{\mathbb{Y}}_2, \hat{h}_1,\hat{h}_2),$$ where $\hat{\mathbb{Y}}_1\subseteq{\mathbb{Y}}_1$. A continuous function $ \mathcal{S}:\mathbb X \times \mathbb{\hat{X}} \to \mathbb{R}_{\geq0} $ is called a storage function from $\hat{\Sigma}$ to $\Sigma$ if  $\forall x\in \mathbb X$ and $\forall \hat{x} \in \mathbb{\hat{X}}$ one has
		\begin{IEEEeqnarray}{c}\label{eq:STFC1}
			\begin{IEEEeqnarraybox}[\relax][c]{rCl}
				\alpha (\Vert h_1(x)-\hat{h}_1(\hat{x})\Vert ) &\leq& \mathcal{S}(x,\hat{x}),
			\end{IEEEeqnarraybox}
		\end{IEEEeqnarray}
		and $\forall x\in \mathbb X$, $\forall \hat x\in \mathbb{\hat{X}}$, $\forall \hat u\in\mathbb{\hat{U}}$, $\exists u\in\mathbb U$, $\forall w\in\mathbb W$, $\forall \hat w\in\mathbb{\hat{W}}$, $\forall x_d \in f(x,u,w), $ $\exists\hat{x}_{d} \in \hat{f}(\hat{x},\hat{u},\hat{w})$ such that one gets
		\begin{IEEEeqnarray}{c}\label{eq:STFC2}
			\begin{IEEEeqnarraybox}[\relax][c]{rCl}
				&&\mathcal{S}(x_d,\hat{x}_d)-\mathcal{S}(x,\hat{x})
				\leq-\sigma(\mathcal{S}(x,\hat{x}))+\rho_{ext}(\Vert \hat{u}\Vert )+\\
				&&\begin{bmatrix}
					Ww-\hat W\hat w\\
					h_2(x)-H\hat h_2(\hat x)
				\end{bmatrix}^T\overbrace{\begin{bmatrix}
						X^{11}&X^{12}\\
						X^{21}&X^{22}
				\end{bmatrix}}^{X:=}
				\begin{bmatrix}
					Ww-\hat W\hat w\\
					h_2(x)-H\hat h_2(\hat x)
				\end{bmatrix}+\epsilon,
			\end{IEEEeqnarraybox}
		\end{IEEEeqnarray}
		for some $\alpha,\sigma \in \mathcal{K}_{\infty}$, $\rho_{ext}  \in \mathcal{K}_{\infty}\cup \{0\} ,$ some matrices $W,\hat W,H$ of appropriate dimensions, some symmetric matrix $X$ of appropriate dimension with conformal block partitions $X^{ij}$, $i,j\in[1;2]$,  and some $\epsilon \in \mathbb{R}_{\geq 0}$. Here, system $\hat{\Sigma}$ is called an abstraction of $\Sigma$. 
	\end{definition}
	Note that $\hat{\Sigma}$ may be finite or infinite depending on cardinalities of sets $\mathbb{\hat{X}},\mathbb{\hat{U}},\mathbb{\hat{W}}$. Now, we define a notion of so-called simulation function, inspired by Definition 1 in \cite{Girard2009566}, which quantifies the error between systems $\Sigma$ and $\hat{\Sigma}$ both without internal inputs and outputs.
	\begin{definition}\label{def:SFD}
		Consider systems $\Sigma=(\mathbb X,\mathbb U,\mathcal{U},f,\mathbb Y,h)$ and  $\hat{\Sigma}=(\hat{\mathbb{X}},\hat{\mathbb{U}},\hat{\mathcal{U}},\hat{f},\hat{\mathbb{Y}},\hat{h})$. A continuous function $ V:\mathbb X\times \mathbb{\hat{X}} \to \mathbb{R}_{\geq0} $ is called a  simulation function from $\hat{\Sigma}$ to $\Sigma$ if $\forall x\in \mathbb X$ and $\forall \hat x\in\mathbb{\hat{X}}$ one has
		\begin{IEEEeqnarray}{c}\label{e:SFC1}
			\begin{IEEEeqnarraybox}[\relax][c]{rCl}
				\alpha (\Vert h(x)-\hat{h}(\hat{x})\Vert ) &\leq& V(x,\hat{x}),
			\end{IEEEeqnarraybox}
		\end{IEEEeqnarray}
		and $\forall x\in \mathbb X $, $\forall \hat x\in \mathbb{\hat{X}}$, $\forall \hat u\in\mathbb{\hat{U}}$, $\exists u\in\mathbb U$, $\forall x_d \in f(x,u)$, $\exists\hat{x}_{d} \in \hat{f}(\hat{x},\hat{u})$ such that one gets
		\begin{align}\label{e:SFC2}
		V(x_d,\hat{x}_d)-V(x,\hat{x})
		\leq -\sigma(V(x,\hat{x}))+\rho_{ext}(\Vert \hat{u}\Vert )+\varepsilon,
		\end{align}
		for some $\alpha , \sigma \in \mathcal{K}_{\infty}$, $\rho_{ext} \in \mathcal{K}_{\infty}\cup \{0\} $, and some $\varepsilon \in \mathbb{R}_{\geq 0}$.
	\end{definition}
	We say that a system $\hat{\Sigma} $ is approximately alternatingly simulated by a system $\Sigma $ or a system $\Sigma $ approximately alternatingly simulates a system $\hat{\Sigma} $, denoted by $\hat{\Sigma} \preceq _{\mathcal{AS}}  \Sigma$, if there
	exists a simulation function from $\hat{\Sigma} $ to $\Sigma $ as in Definition \ref{def:SFD}.
	
	%%*************************************************************************
	In general the notions of storage functions in Definition \ref{def:DS}
	and simulation functions in Definition \ref{def:SFD} are not comparable. The former is established for systems with
	internal inputs and outputs while the latter is established only for
	systems without internal inputs and outputs. One can
	simply verify that both notions coincide for systems
	without internal inputs and outputs.
	%%*************************************************************************

	Before we provide our first main result, we recall Lemma B.1 in \cite{Jiang2001857}  which is used later to show some of the results.
	\begin{lemma}\label{lemma1}
		For any function $\sigma\in\mathcal{K}_{\infty}$, there exists a function $\hat{\sigma}\in\mathcal{K}_{\infty}$ satisfying $\hat{\sigma}(s) \leq \sigma(s)$ for all $s\in \R_{\ge0}$, and\footnote{Here, $\mathcal{I}_d\in\mathcal{K}_{\infty}$ denotes the identity function.}  $\mathcal{I}_d-\hat{\sigma} \in \mathcal{K}$.
	\end{lemma}
	The next theorem shows the importance of the existence of a simulation function.
	\begin{theorem}\label{thm:1}
		Consider systems $\Sigma=(\mathbb X,\mathbb U,\mathcal{U},f,\mathbb Y,h)$ and $\hat{\Sigma}=(\mathbb{\hat{X}},\mathbb{\hat{U}},\hat{\mathcal{U}},\hat{f},\hat{\mathbb{Y}}, \hat{h})$. Suppose $V$ is a simulation function from $\hat{\Sigma}$ to $\Sigma$. Then there exist a constant $\varphi \in \mathbb{R}_{\geq 0}$, functions $\gamma_{ext} \in\mathcal{K}_{\infty}\cup\{0\}$ and $\lambda\in\mathcal{K}_\infty$, where $\lambda(s)<s$ $\forall s\in\R_{\geq0}$, such that for any $x\in \mathbb X$, $\hat{x} \in \mathbb{\hat{X}}$,  $\hat{u} \in \hat{\mathbb{U}}$, there exits $ u \in \mathbb{U} $ so that for any $x_d \in f(x,u)$ in $\Sigma$ there exists $\hat{x}_{d} \in \hat{f}(\hat{x},\hat{u})$ in $\hat{\Sigma}$ such that
		\begin{align}\label{eq:11}
		\alpha(\Vert h(x_d)-\hat{h}(\hat{x}_{d})\Vert )\leq& V(x_d,\hat{x}_d)\\\notag\leq&
		\max\{\lambda(V(x,\hat{x})),\gamma_{ext}(\Vert \hat{u}\Vert )+\varphi\}.
		\end{align}
	\end{theorem}
	\begin{proof}
		Let $\psi$ be a $\mathcal{K}_{\infty}$ function such that $\mathcal{I}_d-\psi \in \mathcal{K}_{\infty}$ and define
		$c=\hat{\sigma}^{-1}(\psi^{-1})(\rho_{ext}(\Vert \hat{u}\Vert) +\varepsilon)$, where $\hat\sigma$ is given as in Lemma \ref{lemma1} for function $\sigma$ appearing in Definition \ref{def:SFD}.  
		Let 
		$D=\{(x,\hat{x}) \in \mathbb X \times \mathbb{\hat{X}} | V(x,\hat{x}) \leq c\}$.
		First, assume $ (x, \hat{x}) \in D $. 
		Then $ V(x, \hat{x})\leq c$, that is, $\psi (\hat{\sigma}( V(x, \hat{x})\leq \rho_{ext}(\Vert \hat{u}\Vert )+\varepsilon$. Since $(\mathcal{I}_d-\hat{\sigma}) \in \mathcal{K}$, and $\psi(\hat{\sigma}(c))=\rho_{ext}(\Vert \hat{u}\Vert )+\varepsilon$, and by using \eqref{e:SFC2}, one obtains
		\begin{IEEEeqnarray}{c}\notag
			\begin{IEEEeqnarraybox}[\relax][c]{rCl}
				V(x_d, \hat{x}_{d})&\leq& V(x,\hat{x}) -\hat{\sigma}(V(x,\hat{x}))+\rho_{ext}(\Vert \hat{u}\Vert )+\varepsilon\\
				&\leq& (\mathcal{I}_d-\hat{\sigma})(V(x,\hat{x})))+\rho_{ext}(\Vert \hat{u}\Vert )+\varepsilon\\
				&\leq&(\mathcal{I}_d-\hat{\sigma})(c)+\psi(\hat{\sigma}(c))\\
				&\leq&  c-\hat{\sigma}(c)+\psi(\hat{\sigma}(c))\\
				&\leq&-(\mathcal{I}_d- \psi)(\hat{\sigma}(c))+c\leq c,
			\end{IEEEeqnarraybox}
		\end{IEEEeqnarray}
		for all $x_d \in f(x,u)$ and some $\hat{x}_{d} \in \hat{f}(\hat{x},\hat{u})$.
		Using the definition of $c$, we have the following inequality		
		\begin{align}\label{eq:12}
		V(x_d, \hat{x}_d)\leq\hat{\sigma}^{-1}(\psi^{-1}(\rho_{ext}(\Vert \hat{u}\Vert )+\varepsilon)
		\leq\gamma_{ext}(\Vert\hat{u}\Vert )+\varphi,
		\end{align}		
		where $\gamma_{ext}(s)=\hat{\sigma}^{-1}(\psi^{-1}(2\rho_{ext}(s)))$ $\forall s\in\R_{\geq0}$, and $\varphi=\hat{\sigma}^{-1}(\psi^{-1}(2\varepsilon))$.
		Now assume $ (x, \hat{x}) \notin D $. Then $\psi(\hat{\sigma}(V(x,\hat{x})))\geq \rho_{ext}(\Vert \hat{u}\Vert )+\varepsilon$, and one has
		\begin{IEEEeqnarray}{c}\label{ine1}
			\begin{IEEEeqnarraybox}[\relax][a]{rCl}
				V(x_d, \hat{x}_{d})&\leq&V(x,\hat{x}) -\hat{\sigma}(V(x,\hat{x}))+\psi(\hat{\sigma}(V(x,\hat{x}))\\
				&\leq&V(x,\hat{x}) - (\mathcal{I}_d - \psi)(\hat{\sigma}(V(x,\hat{x})))\\
				&\leq& - \tilde{\psi}(V(x,\hat{x})) +V(x,\hat{x})\\
				&\leq& (\mathcal{I}_d - \tilde{\psi})(V(x,\hat{x})),
			\end{IEEEeqnarraybox}
		\end{IEEEeqnarray}
		for all $x_d \in f(x,u)$ and some $\hat{x}_{d} \in \hat{f}(\hat{x},\hat{u})$, where $\tilde{\psi}(s):=(\mathcal{I}_d -\psi)(\hat{\sigma}(s))$ $\forall s\in\R_{\geq0}$.
		Observe that $(\mathcal{I}_d - \tilde{\psi})$ is a $  \mathcal{K}_\infty$ function since $\mathcal{I}_d -\psi$ and $\hat{\sigma}$ are $\mathcal{K}_\infty$ functions and $(\mathcal{I}_d - \tilde{\psi})(s)<s$ $\forall s\in\R_{\geq0}$.
		From \eqref{ine1} and by defining $\lambda(s)=(\mathcal{I}_d - \tilde{\psi})(s)$ $\forall s\in\R_{\geq0}$, one gets
		\begin{IEEEeqnarray}{c}\label{eq:13}
			\begin{IEEEeqnarraybox}[\relax][a]{rCl}
				V(x_d,\hat{x}_{d})\leq  \lambda(V(x,\hat{x})).
			\end{IEEEeqnarraybox}
		\end{IEEEeqnarray}
		Combining \eqref{eq:12} and \eqref{eq:13}, and by using \eqref{e:SFC1}, one gets
		\begin{align}\notag
		\alpha(\Vert h(x_d)-\hat{h}(\hat{x}_{d})\Vert )&\leq V(x_d,\hat{x}_d)\\\notag\leq& \max\{\lambda(V(x,\hat{x})),\gamma_{ext}(\Vert \hat{u}\Vert )+\varphi\},
		\end{align}which completes the proof.
	\end{proof}
		\begin{remark}
			Assume that $\exists v\in \R_{>0}$ such that $\Vert \hat{u} \Vert \leq v$ $ \forall \hat{u} \in \mathbb{\hat{U}}$. Then Theorem \ref{thm:1} implies that the relation $R\subseteq\mathbb{X}\times \hat{\mathbb{X}}$ defined by $R=\{(x,\hat{x})\in \mathbb{X}\times \hat{\mathbb{X}}|V(x,\hat{x})\leq \gamma_{ext}(v)+\varphi\}$ is an $\hat\varepsilon$-approximate alternating simulation relation, defined in \cite{Tabu}, from $\hat{\Sigma}$ to $\Sigma$ with $\hat\varepsilon=\alpha^{-1}(\gamma_{ext}(v)+\varphi)$.
		\end{remark}
				
	\section{Compositionality Result}\label{1:III}
	\label{s:inter}
	In this section, we analyze networks of discrete-time control systems and show
	how to construct a simulation function from a network of their abstractions to the concrete network by using storage functions of the subsystems. The
	definition of the network of discrete-time control systems is based on the notion of
	interconnected systems described in~\cite{murat}. 
	
	\subsection{Interconnected systems}
	
	Here, we define the \emph{interconnected discrete-time control system} as the following.
	
	\begin{definition}\label{d:ics}
		Consider $N\in\N_{\geq1}$ discrete-time control subsystems $
		\Sigma_i=\left(\mathbb X_i,\mathbb U_i,\mathbb W_i,\mathcal{U}_i,\mathcal{W}_i,f_i,\mathbb Y_{1i},\mathbb Y_{2i},h_{1i},h_{2i}\right)$,
		$i\in\intcc{1;N}$, and a static matrix $M$ of an appropriate dimension defining the coupling of these subsystems, where\footnote{This condition is required to have a well-defined interconnection.} $M\prod_{i=1}^N \mathbb Y_{2i} \subseteq \prod_{i=1}^N \mathbb W_{i}$. The \emph{interconnected discrete-time control
			system} $\Sigma=\left(\mathbb X,\mathbb U,\mathcal{U},f,\mathbb Y,h\right)$, denoted by
		$\mathcal{I}(\Sigma_1,\ldots,\Sigma_N)$, follows by $ \mathbb X =\prod_{i=1}^N \mathbb X_i$,
		$ \mathbb U=\prod_{i=1}^N \mathbb U_i$, $ \mathbb Y=\prod_{i=1}^N \mathbb Y_{1i}$, and maps
		\begin{IEEEeqnarray}{c}\notag
			\begin{IEEEeqnarraybox}[\relax][c]{rCl}
				f(x,u)&&\Let\{\intcc{x'_1;\ldots;x'_N}\,\,|\,\,x'_i\in f_i(x_i,u_i,w_i)~ \forall i\in[1;N]\},\\
				h(x)&&\Let \intcc{h_{11}(x_1);\ldots;h_{1N}(x_N)},
			\end{IEEEeqnarraybox}
		\end{IEEEeqnarray}
		where $u=\intcc{u_{1};\ldots;u_{N}}$, $x=\intcc{x_{1};\ldots;x_{N}}$ and
		with the internal variables constrained by
		$$\intcc{w_{1};\ldots;w_{N}}=M\intcc{h_{21}(x_1);\ldots;h_{2N}(x_N)}.$$
	\end{definition}
	
%	An interconnection of $N$ control subsystems $\Sigma_i$ is illustrated schematically in Figure \ref{system1}.
%	\begin{figure}[ht]
%		\begin{center}
%			\includegraphics[width=4cm]{network}
%			\caption{An interconnection of $N$ control subsystems $\Sigma_1,\ldots,\Sigma_N$.}
%			\label{system1}
%		\end{center}
%		\vspace{-0.7cm}
%	\end{figure}

	\subsection{Composing simulation functions from storage functions}
	We assume that we are given $N$ control subsystems
	$\Sigma_i=\left(\mathbb X_i,\mathbb U_i,\mathbb W_i,\mathcal{U}_i,\mathcal{W}_i,f_i,\mathbb Y_{1i},\mathbb Y_{2i},h_{1i},h_{2i}\right)$
	together with their abstractions $\hat{\Sigma}_i=(\mathbb{\hat{X}}_i,\mathbb{\hat{U}}_i,\mathbb{\hat{W}}_i,\hat{\mathcal{U}}_i,\hat{\mathcal{W}}_i,\hat{f}_i,\hat{\mathbb{Y}}_{1i},\hat{\mathbb{Y}}_{2i}, \hat{h}_{1i},\hat{h}_{2i})$, and  storage functions $\mathcal{S}_i$ from
	$\hat\Sigma_i$ to $\Sigma_i$. We use $\alpha_{i}$, $\sigma_i$, $\rho_{i{ext}}$, $H_i$, $W_i$, $\hat W_i$, $X_i$, $X_i^{11}$, $X_i^{12}$, $X_i^{21}$, $X_i^{22}, $ and $ \epsilon_i$ to denote the corresponding functions, matrices, and their corresponding conformal block partitions, and constants appearing in Definition \ref{def:DS}.
	
	The next theorem provides a compositional approach on the construction of abstractions of networks of control subsystems and that of the corresponding simulation function.
	\begin{theorem}\label{t:ic}
		Consider the interconnected control system
		$\Sigma=\mathcal{I}(\Sigma_1,\ldots,\Sigma_N)$ induced by $N\in\N_{\geq1}$ control subsystems~$\Sigma_i$ and coupling matrix $M$. Suppose  each control subsystem $\Sigma_i$ admits an abstraction $\hat \Sigma_i$ with the corresponding storage function $\mathcal{S}_i$. If there exist $\mu_{i}\geq0$, $i\in[1;N]$, and matrix $\hat{M}$ of appropriate dimension such that the matrix (in)equality and inclusion
		\begin{IEEEeqnarray}{c}\label{e:MC1}
			\begin{IEEEeqnarraybox}[\relax][c]{rCl}\
				\begin{bmatrix}
					WM\\
					I_{q}
				\end{bmatrix}^T X(\mu_1X_1,\ldots,\mu_NX_N)\begin{bmatrix}
					WM\\
					I_{q}
				\end{bmatrix}\preceq0,
			\end{IEEEeqnarraybox}
		\end{IEEEeqnarray}
		\begin{IEEEeqnarray}{c}\label{e:MC2}
			\begin{IEEEeqnarraybox}[\relax][c]{rCl}\
				WMH=\hat W\hat M,
			\end{IEEEeqnarraybox}
		\end{IEEEeqnarray}
		\begin{IEEEeqnarray}{c}\label{e:MC3}
			\begin{IEEEeqnarraybox}[\relax][c]{rCl}\
				\hat M\prod_{i=1}^N \mathbb{\hat Y}_{2i} \subseteq \prod_{i=1}^N \mathbb{\hat W}_{i},
			\end{IEEEeqnarraybox}
		\end{IEEEeqnarray}
		are satisfied, where
		\begin{IEEEeqnarray}{c}\label{matrix1}
			\begin{IEEEeqnarraybox}[\relax][c]{rCl}
				&&W:=\mathsf{diag}(W_1,\ldots,W_N),\hat W:=\mathsf{diag}(\hat W_1,\ldots,\hat W_N)\\
				&&H:=\mathsf{diag}(H_1,\ldots,H_N),
			\end{IEEEeqnarraybox}
		\end{IEEEeqnarray}
		\begin{IEEEeqnarray}{c}\label{matrix}
			\begin{IEEEeqnarraybox}[\relax][c]{rCl}
				&&X(\mu_1X_1,\ldots,\mu_NX_N):=\\
				&&\begin{bmatrix}
					\mu_1X_1^{11}&&&\mu_1X_1^{12}&&\\
					&\ddots&&&\ddots&\\
					&&\mu_NX_N^{11}&&&\mu_NX_N^{12}\\
					\mu_1X_1^{21}&&&\mu_1X_1^{22}&&\\
					&\ddots&&&\ddots&\\
					&&\mu_NX_N^{21}&&&\mu_NX_N^{22}
				\end{bmatrix},
			\end{IEEEeqnarraybox}
		\end{IEEEeqnarray}
		and $q$ is the number of rows in $H$, then 
		\begin{IEEEeqnarray*}{c}
			V(x,\hat x)\Let\sum_{i=1}^N\mu_i\mathcal{S}_i(x_i,\hat x_i)
		\end{IEEEeqnarray*}
		is a simulation
		function from $\hat\Sigma=\mathcal{I}(\hat\Sigma_1,\ldots,\hat\Sigma_N)$, with the coupling matrix $\hat M$, to $\Sigma$. 
	\end{theorem}
	\begin{proof}
		First we show that inequality \eqref{e:SFC1} holds for some
		$\mathcal{K}_\infty$ function $\alpha$. For any
		$x=\intcc{x_1;\ldots;x_N}\in\mathbb X$ and 
		$\hat x=\intcc{\hat x_1;\ldots;\hat x_N}\in\hat{\mathbb X}$, one gets:
		\begin{IEEEeqnarray}{c}\notag
			\begin{IEEEeqnarraybox}[\relax][c]{rCl}
				\Vert h(x)&&-\hat h(\hat x)\Vert \\
				=&&\Vert [h_{11}(x_1);\ldots;h_{1N}(x_N)]-[\hat h_{11}(\hat x_1);\ldots;\hat h_{1N}(\hat x_N)]\Vert \\\notag
				\le&&\sum_{i=1}^N \Vert h_{1i}(x_i)-\hat h_{1i}(\hat{x}_i)\Vert 
				\le \sum_{i=1}^N \alpha_{i}^{-1}(\mathcal{S}_i( x_i, \hat x_i))\\\leq&& \overline{\alpha}\big(V( x, \hat x)\big),
			\end{IEEEeqnarraybox}
		\end{IEEEeqnarray}	
		where
		$\overline\alpha$ is a $\mathcal{K}_\infty$ function defined as
		\begin{IEEEeqnarray}{c}\notag
			\begin{IEEEeqnarraybox}[\relax][c]{rCl}
				\overline\alpha(s)&=&\max\limits_{\hat{s}\geq 0}\Bigg\{\sum_{i=1}^N\alpha^{-1}_i(s_i)|\mu^T\hat{s}=s\Bigg\},
			\end{IEEEeqnarraybox}
		\end{IEEEeqnarray}
		where $\hat{s}=\intcc{s_1;\ldots;s_N}\in\R^N$ and $\mu=\intcc{\mu_1;\ldots;\mu_N}$. By defining the
		$\mathcal{K}_\infty$ function $\alpha(s)=\overline\alpha^{-1}(s)$, $\forall s\in\R_{\ge0}$, one obtains
		$$\alpha(\Vert h(x)-\hat h(\hat x)\Vert )\le V( x, \hat x),$$
		satisfying inequality \eqref{e:SFC1}.
		Now we show that inequality \eqref{e:SFC2} holds as well.
		Consider any 
		$x=\intcc{x_1;\ldots;x_N}\in\mathbb X$,
		$\hat x=\intcc{\hat x_1;\ldots;\hat x_N}\in\hat{\mathbb{X}}$, and
		$\hat u=\intcc{\hat u_{1};\ldots;\hat u_{N}}\in\hat{\mathbb{U}}$. For any $i\in[1;N]$, there exists $u_i\in\mathbb U_i$, consequently, a vector $u=\intcc{u_{1};\ldots;u_{N}}\in\mathbb U$ such that  for any $x_d \in f(x,u)$ there exists  $\hat{x}_{d} \in \hat{f}(\hat{x},\hat{u})$ satisfying \eqref{eq:STFC2} for each pair of subsystems $\Sigma_i$ and $\hat\Sigma_i$
		with the internal inputs given by $\intcc{w_1;\ldots;w_N}=M[h_{21}(x_1);\ldots;h_{2N}(x_N)]$ and $\intcc{\hat w_1;\ldots;\hat w_N}=\hat M[\hat h_{21}(\hat x_1);\ldots;\hat h_{2N}(\hat x_N)]$.
		We derive the following inequality  
		\begin{IEEEeqnarray}{c}\label{rewrite}
			\begin{IEEEeqnarraybox}[\relax][c]{rCl}
				V(x_d,\hat{x}_{d})&&-V(x,\hat{x})\\
				=&&\sum_{i=1}^N \mu_i\big(\mathcal{S}_i(x_{d_i},\hat{x}_{d_i})-\mathcal{S}_i(x_i,\hat{x}_i)\big)\\
				\leq&&\sum_{i=1}^N\mu_i\bigg(-\sigma_i(\mathcal{S}_i(x_i,\hat{x}_i))+\rho_{iext}(\Vert \hat{u}_i\Vert )+\\
				&&\begin{bmatrix}
					W_iw_i-\hat W_i\hat w_i\\
					h_{2i}(x_i)-H_i\hat h_{2i}(\hat x_i)
				\end{bmatrix}^T\overbrace{\begin{bmatrix}
						X_i^{11}&X_i^{12}\\
						X_i^{21}&X_i^{22}
				\end{bmatrix}}^{X_i:=}\\
				&&\begin{bmatrix}
					W_iw_i-\hat W_i\hat w_i\\
					h_{2i}(x_i)-H_i\hat h_{2i}(\hat x_i)
				\end{bmatrix}+\epsilon_i\bigg).
			\end{IEEEeqnarraybox}
		\end{IEEEeqnarray}
		Using conditions \eqref{e:MC1}, \eqref{e:MC2}, and the definition of matrices $W$, $\hat W$, $H$, and $X$ in \eqref{matrix1} and \eqref{matrix}, the inequality \eqref{rewrite} can be rewritten as
		\begin{IEEEeqnarray}{c}\notag
			\begin{IEEEeqnarraybox}[\relax][c]{rCl}
				V(&&x_d,\hat{x}_{d})-V(x,\hat{x})\\
%				\leq&&\sum_{i=1}^N\mu_i\big(-\sigma_i(\mathcal{S}_i(x_i,\hat{x}_i))\big)+
%				\sum_{i=1}^N\mu_i\rho_{iext}(\Vert \hat{u}_i\Vert )+\sum_{i=1}^N\mu_i \epsilon_i\\
%				&&+\begin{bmatrix}
%					W_1w_1-\hat W_1\hat w_1\\
%					\vdots\\
%					W_Nw_N-\hat W_N\hat w_N\\
%					h_{21}(x_1)-H_1\hat h_{21}(\hat x_1)\\
%					\vdots\\
%					h_{2N}(x_N)-H_N\hat h_{2N}(\hat x_N)
%				\end{bmatrix}^TX(\mu_1X_1,\ldots,\mu_NX_N)\\
%				&&~~\begin{bmatrix}
%					W_1w_1-\hat W_1\hat w_1\\
%					\vdots\\
%					W_Nw_N-\hat W_N\hat w_N\\
%					h_{21}(x_1)-H_1\hat h_{21}(\hat x_1)\\
%					\vdots\\
%					h_{2N}(x_N)-H_N\hat h_{2N}(\hat x_N)
%				\end{bmatrix}\\\notag
				\leq&&\sum_{i=1}^N\mu_i\big(-\sigma_i(\mathcal{S}_i(x_i,\hat{x}_i))\big)+\sum_{i=1}^N\mu_i\rho_{iext}(\Vert \hat{u}_i\Vert )+\sum_{i=1}^N\mu_i \epsilon_i\\
				&&+\begin{bmatrix}
					W\begin{bmatrix}
						w_1\\
						\vdots\\
						w_N
					\end{bmatrix}-\hat W\begin{bmatrix}
						\hat w_1\\
						\vdots\\
						\hat w_N
					\end{bmatrix}\\
					h_{21}(x_1)-H_1\hat h_{21}(\hat x_1)\\
					\vdots\\
					h_{2N}(x_N)-H_N\hat h_{2N}(\hat x_N)
				\end{bmatrix}^TX(\mu_1X_1,\ldots,\mu_NX_N)\\
				&&~~\begin{bmatrix}
					W\begin{bmatrix}
						w_1\\
						\vdots\\
						w_N
					\end{bmatrix}-\hat W\begin{bmatrix}
						\hat w_1\\
						\vdots\\
						\hat w_N
					\end{bmatrix}\\
					h_{21}(x_1)-H_1\hat h_{21}(\hat x_1)\\
					\vdots\\
					h_{2N}(x_N)-H_N\hat h_{2N}(\hat x_N)
				\end{bmatrix}\\\notag
				\leq&&\sum_{i=1}^N\mu_i\big(-\sigma_i(\mathcal{S}_i(x_i,\hat{x}_i))\big)+\sum_{i=1}^N\mu_i\rho_{iext}(\Vert \hat{u}_i\Vert )+\sum_{i=1}^N\mu_i \epsilon_i\\
				&&+\begin{bmatrix}
					h_{21}(x_1)-H_1\hat h_{21}(\hat x_1)\\
					\vdots\\
					h_{2N}(x_N)-H_N\hat h_{2N}(\hat x_N)
				\end{bmatrix}^T\\
				&&~~\begin{bmatrix}
					WM\\
					I_{q}
				\end{bmatrix}^TX(\mu_1X_1,\ldots,\mu_NX_N)\begin{bmatrix}
					WM\\
					I_{q}
				\end{bmatrix}\\
				&&~~\begin{bmatrix}
					h_{21}(x_1)-H_1\hat h_{21}(\hat x_1)\\
					\vdots\\
					h_{2N}(x_N)-H_N\hat h_{2N}(\hat x_N)
				\end{bmatrix}\\\notag
				\leq&&\sum_{i=1}^N\mu_i\big(-\sigma_i(\mathcal{S}_i(x_i,\hat{x}_i))\big)+\sum_{i=1}^N\mu_i\rho_{iext}(\Vert \hat{u}_i\Vert )+\sum_{i=1}^N\mu_i \epsilon_i.
			\end{IEEEeqnarraybox}
		\end{IEEEeqnarray}
		Remark that condition \eqref{e:MC3} ensures that the interconnection $\hat\Sigma=\mathcal{I}(\hat\Sigma_1,\ldots,\hat\Sigma_N)$ is well-defined. By defining 
		\begin{IEEEeqnarray}{c}\notag
			\begin{IEEEeqnarraybox}[\relax][c]{rCl}
				\sigma(s)&&\Let\min\limits_{\hat{s}\geq 0}\Bigg\{\sum_{i=1}^N\mu_i\sigma_i(s_i)|\mu^T\hat{s}=s\Bigg\},\varepsilon\Let\sum_{i=1}^N\mu_i \epsilon_i,\\
				\rho_{ext}(s)&&\Let\max\limits_{\hat{s}\geq 0}\Bigg\{\sum_{i=1}^N\mu_i\rho_{iext}(\Vert {s_i}\Vert )\,\,|\,\,\Vert \hat{s}\Vert =s\Bigg\}, 				
			\end{IEEEeqnarraybox}
		\end{IEEEeqnarray}
		where $\sigma\in\mathcal{K}_\infty$ and $\rho_{\mathrm{ext}}\in\mathcal{K}_\infty\cup\{0\}$,
		we readily have
		\begin{IEEEeqnarray}{c}\notag
			\begin{IEEEeqnarraybox}[\relax][c]{rCl}
				V&&(x_d,\hat{x}_d)-V(x,\hat{x})\leq-\sigma\left(V\left( x,\hat{x}\right)\right)+\rho_{\mathrm{ext}}(\Vert \hat{u}\Vert )+\varepsilon,
			\end{IEEEeqnarraybox}
		\end{IEEEeqnarray}
		which satisfies inequality \eqref{e:SFC2}. Hence, $V$ is a simulation function from $\hat \Sigma$ to $\Sigma$. 
	\end{proof}

	\section{Construction of Symbolic Models}\label{1:IV}
	
	In the previous sections, $\Sigma$ and $\hat{\Sigma}$ were considered as general discrete-time control systems, deterministic or nondeterministic, finite or infinite, that can be related to each other through a storage function  or a simulation function when $\Sigma$ and $\hat{\Sigma}$ are networks of control subsystems. In this section, we consider $\Sigma$ as an infinite, deterministic, control system and $\hat{\Sigma}$ as its finite abstraction.
	In addition, the storage function from $\hat{\Sigma}$ to $\Sigma$ is established under the assumption that the original discrete-time control system $\Sigma$ is so-called incrementally passivable.
	
	In order to show our next main result we introduce the following definition.
	
	\begin{definition}\label{ass:1}
		Control system 
		$$\Sigma=(\mathbb X,\mathbb U,\mathbb W,\mathcal{U},\mathcal{W},f,\mathbb X,\mathbb Y_{2},1_{\mathbb X,},h_2),$$
		is called incrementally passivable if there exist functions $\mathcal{H}:\mathbb X \to \mathbb U$ and  $ \mathcal{G}:\mathbb X \times \mathbb X \to \mathbb{R}_{\geq0} $  such that $\forall x,x'\in \mathbb X$, $\forall u\in \mathbb U$, $\forall w,w' \in \mathbb W$, the inequalities:
		\begin{IEEEeqnarray}{c}\label{eq:ISTFC1}
			\begin{IEEEeqnarraybox}[\relax][c]{rCl}
				\underline{\alpha} (\Vert x-x'\Vert ) &\leq& \mathcal{G}(x,x'),%\leq \overline{\alpha}(\Vert x-x'\Vert ),
			\end{IEEEeqnarraybox}
		\end{IEEEeqnarray}
		and
		\begin{IEEEeqnarray}{c}\label{eq:ISTFC2}
			\begin{IEEEeqnarraybox}[\relax][c]{rCl}
				\mathcal{G}&&(f(x,\mathcal{H}(x)+u,w),f(x',\mathcal{H}(x')+u,w'))-\mathcal{G}(x,x')\\
				\leq&&-\hat{\kappa}(\mathcal{G}(x,x'))+\\
				&&\begin{bmatrix}w-w'\\
					h_2(x)-h_2(x')
				\end{bmatrix}^T\overbrace{\begin{bmatrix}
						X^{11}&X^{12}\\
						X^{21}&X^{22}
				\end{bmatrix}}^{X:=}\begin{bmatrix}
					w-w'\\
					h_2(x)-h_2(x')
				\end{bmatrix},
			\end{IEEEeqnarraybox}
		\end{IEEEeqnarray}
		hold for some $\underline{\alpha}, \hat{\kappa} \in \mathcal{K}_{\infty} $, and matrix $X$ of appropriate dimension.
		
		\begin{remark}
			Note that any stabilizable linear control system is incrementally passivable as in Definition \ref{ass:1}. Moreover, any incrementally input-to-state stabilizable control system as in Definition \ref{ass:2} with $\rho_{int}(r)=cr^2$, for some $c\in \R_{>0}$ and any $r\in \R_{\ge0}$, is also incrementally passivable. 
		\end{remark}
		
		\begin{definition}\label{ass:2}
			Control system 
			$$\Sigma=(\mathbb X,\mathbb U,\mathbb W,\mathcal{U},\mathcal{W},f,\mathbb X,\mathbb Y_{2},1_{\mathbb X,},h_2)$$
			is called incrementally input-to-state stabilizable with respect to the internal input if there exist functions $\mathcal{H}:\mathbb X \to \mathbb U$ and  $ \hat{V}:\mathbb X \times \mathbb X \to \mathbb{R}_{\geq0} $  such that $\forall x,x'\in \mathbb X$, $\forall u\in \mathbb U$, $\forall w,w' \in \mathbb W$, the inequalities:
			\begin{IEEEeqnarray}{c}\label{eq:ISTFC3}
				\begin{IEEEeqnarraybox}[\relax][c]{rCl}
					\underline{\alpha} (\Vert x-x'\Vert ) &\leq& \hat{V}(x,x')\leq \overline{\alpha}(\Vert x-x'\Vert ),
				\end{IEEEeqnarraybox}
			\end{IEEEeqnarray}
			and
			\begin{IEEEeqnarray}{c}\label{eq:ISTFC4}
				\begin{IEEEeqnarraybox}[\relax][c]{rCl}
					\hat{V}&&(f(x,\mathcal{H}(x)+u,w),f(x',\mathcal{H}(x')+u,w'))\\-&&\hat{V}(x,x')
					\leq-\hat{\kappa}(\hat{V}(x,x'))+\rho_{int}(\Vert w-w'\Vert),
				\end{IEEEeqnarraybox}
			\end{IEEEeqnarray}
			hold for some $\underline{\alpha}, \overline{\alpha}, \hat{\kappa} \in \mathcal{K}_{\infty}$, and $\rho_{int} \in \mathcal{K}_{\infty}$.
		\end{definition}
		We refer interested readers to \cite{ruffer} for detailed information on incremental stability of discrete-time control systems. 	
	\end{definition}
	Now, we construct a finite abstraction $\hat\Sigma$ of an incrementally passivable control system $\Sigma$.
	\begin{definition}\label{def:sym}
		Given an incrementally passivable control system $\Sigma=(\mathbb X,\mathbb U,\mathbb W,\mathcal{U},\mathcal{W},f,\mathbb X,\mathbb Y_{2},1_{\mathbb X,},h_2)$, where $\mathbb X,\mathbb U,\mathbb W$ are assumed to be finite unions of boxes, one can construct a finite system   
		\begin{IEEEeqnarray}{c}\label{eq:14}
			\begin{IEEEeqnarraybox}[\relax][c]{rCl}
				\hat{\Sigma}&=&(\mathbb{\hat{X}},\mathbb{\hat{U}},\mathbb{\hat{W}},\hat{\mathcal{U}},\hat{\mathcal{W}},\hat{f},\hat{\mathbb{Y}}_1,\hat{\mathbb{Y}}_2, \hat{h}_1,\hat{h}_2),
			\end{IEEEeqnarraybox}
		\end{IEEEeqnarray}
		where:
		\begin{itemize}
			\item $\mathbb{\hat{X}}=[\mathbb X]_\eta$, where $0<\eta\leq\emph{span}(\mathbb X)$ is a state set quantization parameter; 
			\item $\mathbb{\hat{U}}=[\mathbb U]_{\mu_1}$, where $0<\mu_{1}\leq\emph{span}(\mathbb U)$ is the external input set quantization parameter;
			\item $\hat{x}_{d}\in\hat{f}(\hat{x},\hat{u},\hat{w})$ iff $\Vert\hat{x}_{d}-f(\hat{x},\mathcal{H}(\hat{x})+\hat{u},\hat{w})\Vert \leq\eta/2$;
			\item $\hat{\mathbb{Y}}_1=\mathbb{\hat{X}} $, $\hat{\mathbb{Y}}_2=\{ h_2({\hat{x})| \hat{x} \in \hat{\mathbb X}}\}$, $\hat{h}_1=1_{\hat{X}}$ and $\hat{h}_2=h_2 $;
			\item $\mathbb{\hat{W}}$ is an appropriate finite internal input set satisfying condition \eqref{e:MC3} in the compositional setting (cf. the example section).
		\end{itemize}
		
	\end{definition}
	
	Now, we present one of the main results of the paper establishing the relation between $\Sigma$ and $\hat{\Sigma}$, introduced above, via the notion of storage function.  
	\begin{theorem}\label{thm:2}
		Let $\Sigma$ be an incrementally passivable control system as in Definition \ref{ass:1} and $\hat{\Sigma}$ be a finite system as in Definition \ref{def:sym}. Assume that there exists a function $\gamma\in\mathcal{K}_{\infty}$  such that for any $x,x',x'' \in \mathbb{X}$ one has
		\begin{IEEEeqnarray}{c}\label{eq:TI}
			\begin{IEEEeqnarraybox}[\relax][c]{rCl}
				\mathcal{G}(x,x')-\mathcal{G}(x,x'')\leq \gamma(\Vert x'-x''\Vert),
			\end{IEEEeqnarraybox}
		\end{IEEEeqnarray}
		for $\mathcal G$ as in Definition \ref{ass:1}.
		Then $\mathcal{G}$ is a storage function from $\hat{\Sigma}$ to $\Sigma$.
	\end{theorem}
	
	\begin{proof} Since system $\Sigma$ is incrementally passivable, from \eqref{eq:ISTFC1}, $\forall x\in \mathbb{X}$ and $ \forall \hat{x} \in \mathbb{\hat{X}}
		$,  we have 
		\begin{IEEEeqnarray}{c}\notag
			\begin{IEEEeqnarraybox}[\relax][c]{rCl}
				\underline{\alpha} (\Vert x-\hat{x}\Vert )= \underline{\alpha} (\Vert h_1(x)-\hat{h}_1(\hat{x})\Vert ) &\leq& \mathcal{G}(x,\hat{x}),
			\end{IEEEeqnarraybox}
		\end{IEEEeqnarray}	
		satisfying \eqref{eq:STFC1} with  $\alpha(s) \Let \underline{\alpha}(s) $ $\forall s\in \R_{\geq0}$.
		Now from \eqref{eq:TI}, $\forall x\in \mathbb{X}, \forall \hat{x} \in \mathbb{\hat{X}}, \forall \hat{u} \in \mathbb{\hat{U}},\forall w \in \mathbb{W},\forall \hat{w} \in \mathbb{\hat{W}}$, we have 
		\begin{IEEEeqnarray}{c}\notag
			\begin{IEEEeqnarraybox}[\relax][c]{rCl}
				&&\mathcal{G}(f(x,\mathcal{H}(x)+\hat{u},w),\hat{x}_{d})-\\
				&&\mathcal{G}(f(x,\mathcal{H}(x)+\hat{u},w),f(\hat{x},\mathcal{H}(\hat{x})+\hat{u},\hat{w}))\\
				&&\leq\gamma(\Vert\hat{x}_{d}-f(\hat{x},\mathcal{H}(\hat{x})+\hat{u},\hat{w})\Vert),
			\end{IEEEeqnarraybox}
		\end{IEEEeqnarray}
		for any $\hat{x}_{d}\in\hat{f}(\hat{x},\hat{u},\hat{w})$.
		Now, from Definition \ref{def:sym}, the above inequality reduces to
		\begin{IEEEeqnarray}{c}\notag
			\begin{IEEEeqnarraybox}[\relax][c]{rCl}
				\mathcal{G}(f(x&&,\mathcal{H}(x)+\hat{u},w),\hat{x}_{d})\\
				\leq&&\mathcal{G}(f(x,\mathcal{H}(x)+\hat{u},w),f(\hat{x},\mathcal{H}(\hat{x})+\hat{u},\hat{w}))+ \gamma(\eta/2).
			\end{IEEEeqnarraybox}
		\end{IEEEeqnarray}
		Note that by \eqref{eq:ISTFC2} and since $h_2=\hat h_2$, we get 
		\begin{IEEEeqnarray}{c}\notag
			\begin{IEEEeqnarraybox}[\relax][c]{rCl}
				\mathcal{G}&&(f(x,\mathcal{H}(x)+\hat{u},w),f(\hat{x},\mathcal{H}(\hat{x})+\hat{u},\hat{w}))-\mathcal{G}(x,\hat{x})\\
				\leq&&-\hat{\kappa}(\mathcal{G}(x,\hat{x}))+\\
				&&\begin{bmatrix}w-\hat w\\
					h_2(x)-\hat h_2(\hat x)
				\end{bmatrix}^T\begin{bmatrix}
					X^{11}&X^{12}\\
					X^{21}&X^{22}
				\end{bmatrix}\begin{bmatrix}
					w-\hat w\\
					h_2(x)-\hat h_2(\hat x)
				\end{bmatrix}.
			\end{IEEEeqnarraybox}
		\end{IEEEeqnarray}	
		It follows that $\forall x\in \mathbb{X}, \forall \hat{x} \in \mathbb{\hat{X}}, \forall \hat{u} \in \mathbb{\hat{U}},$ and $\forall w \in \mathbb{W},\forall \hat{w} \in \mathbb{\hat{W}}
		$, one obtains
		\begin{IEEEeqnarray}{c}\notag
			\begin{IEEEeqnarraybox}[\relax][c]{rCl}
				&&\mathcal{G}(f(x,\mathcal{H}(x)+\hat{u},w),\hat{x}_{d})-\mathcal{G}(x,\hat{x})\leq-\hat{\kappa}(\mathcal{G}(x,\hat{x}))+\\
				&&\begin{bmatrix}w-\hat w\\
					h_2(x)-\hat h_2(\hat x)
				\end{bmatrix}^T\begin{bmatrix}
					X^{11}&X^{12}\\
					X^{21}&X^{22}
				\end{bmatrix}\begin{bmatrix}
					w-\hat w\\
					h_2(x)-\hat h_2(\hat x)
				\end{bmatrix}+\gamma(\eta/2),\\
			\end{IEEEeqnarraybox}
		\end{IEEEeqnarray}
		for any $\hat{x}_{d}\in\hat{f}(\hat{x},\hat{u},\hat{w})$, satisfying \eqref{eq:STFC2}
		with $\epsilon=\gamma(\eta/2)$, $u=\mathcal{H}(x)+\hat{u}$, $\sigma=\hat{\kappa}$, $\rho_{ext}\equiv 0, $ $W,$ $ \hat{W},$ $H $ are identity matrices of appropriate dimensions. Hence, $\mathcal{G}$ is a storage function from $\hat \Sigma$ to $\Sigma$.  		
	\end{proof}	

	\section{Example}\label{1:V}
	Consider a linear control system $\Sigma$ described by 
	\begin{IEEEeqnarray*}{c}
		\Sigma:\left\{
		\begin{IEEEeqnarraybox}[\relax][c]{rCl}
			\mathbf{x}(k+1)&=&A\mathbf{x}(k)+\nu(k),\\
			\mathbf{y}(k)&=&\mathbf{x}(k),%
		\end{IEEEeqnarraybox}\right.
	\end{IEEEeqnarray*}
	where $A=e^{-L\tau}$ for some matrix $L \in\R^{n\times n}$ and constant $\tau \in \R_{>0}$. Assume $L$ is the Laplacian matrix \cite{Godsil} of an undirected graph.
%	, \eg, for a complete graph:
%	\begin{IEEEeqnarray}{c}\label{complete}
%		\begin{IEEEeqnarraybox}[\relax][c]{rCl}
%			L=\begin{bmatrix}n-1 & -1 & \cdots & \cdots & -1 \\  -1 & n-1 & -1 & \cdots & -1 \\ -1 & -1 & n-1 & \cdots & -1 \\ \vdots &  & \ddots & \ddots & \vdots \\ -1 & \cdots & \cdots & -1 & n-1\end{bmatrix}.
%		\end{IEEEeqnarraybox}
%	\end{IEEEeqnarray}

	We partition $\mathbf{x}(k)$ as $\mathbf{x}(k)=[\mathbf{x}_1(k);\ldots;\mathbf{x}_N(k)]$ and $\nu(k)$ as $\nu(k)=[\nu_1(k);\ldots;\nu_N(k)]$ where $\mathbf{x}_i(k)$ and $\nu_i(k)$ are both taking values in $\mathbb{R}^{n_i}$, $\forall i\in[1;N]$. Now, by introducing $\Sigma_i$ described by
	\begin{IEEEeqnarray*}{c}
		\Sigma_i:\left\{
		\begin{IEEEeqnarraybox}[\relax][c]{rCl}
			\mathbf{x}_i(k+1)&=&\mathbf{x}_i(k)+\omega_i(k)+\nu_i(k),\\
			\mathbf{y}_{1i}(k)&=&\mathbf{x}_i(k),\\
			\mathbf{y}_{2i}(k)&=&\mathbf{x}_i(k),%
		\end{IEEEeqnarraybox}\right.
	\end{IEEEeqnarray*}
	one can readily verify that $\Sigma=\mathcal{I}(\Sigma_1,\ldots,\Sigma_N)$ where the coupling matrix $M$ is given by $M=A-I_n$.
	
	Consider systems $\hat\Sigma_i$ constructed as in Definition \ref{def:sym}. 
	One can readily verify that, for any $i\in[1;N]$, conditions \eqref{eq:ISTFC1} and \eqref{eq:ISTFC2} are satisfied with $\mathcal{G}_i(x_i,\hat{x}_i)=(x_i-\hat x_i)^T(x_i-\hat x_i)$, $\underline{\alpha}_{i}(r)=r^2$, $\hat{\kappa}_i(r)=\lambda r$, $\rho_{iext}(r)\equiv0$, $\forall r\in\R_{\ge0}$, $\mathcal{H}_i(x_i)=-\lambda x_i$, $\forall \lambda \in(0,0.5]$, and 
	\begin{IEEEeqnarray}{c}\label{passivity}
		\begin{IEEEeqnarraybox}[\relax][c]{rCl} 
			X_i=\begin{bmatrix} I_{n_i} & (1-\lambda) I_{n_i} \\ (1-\lambda) I_{n_i} &  0_{n_i} \end{bmatrix}.
		\end{IEEEeqnarraybox}
	\end{IEEEeqnarray}
	Hence, $\mathcal{G}_i(x_i,\hat{x}_i)=(x_i-\hat x_i)^T(x_i-\hat x_i)$ is a storage function from $\hat\Sigma_i$ to $\Sigma_i$.
	
	Now, we look at $\hat\Sigma=\mathcal{I}(\hat\Sigma_1,\ldots,\hat\Sigma_N)$ with a coupling matrix $\hat M$ satisfying condition \eqref{e:MC2} as follows:
	\begin{IEEEeqnarray}{c}\label{graph}
		\begin{IEEEeqnarraybox}[\relax][c]{rCl} 
			A-I_n=\hat M.
		\end{IEEEeqnarraybox}
	\end{IEEEeqnarray}
	
	Choosing $\mu_1=\cdots=\mu_N=1$ and using $X_i$ in (\ref{passivity}), matrix $X$ in \eqref{matrix} reduces to
	$$
	X=\begin{bmatrix} I_{n} & (1-\lambda)I_n \\ (1-\lambda)I_n &  0_n \end{bmatrix},
	$$
	and condition \eqref{e:MC1} reduces to
	$$
	\begin{bmatrix} A-I_n \\ I_n \end{bmatrix}^T\hspace{-3pt}X\hspace{-2pt}\begin{bmatrix} A-I_n \\ I_n \end{bmatrix}=(A-I_n)\Big(A-I_n+2(1-\lambda)I_n\Big)\preceq 0,
	$$
	which always holds without any restrictions on the number of the subsystems. In order to show the above inequality, we used $A^T=A\succeq0$, $A-I_n\preceq0$, and $2(1-\lambda)I_n +A-I_n\succeq0 $. Note that by choosing finite internal input sets $\mathbb{\hat{W}}_i$ of $\hat\Sigma_i$ in such a way that $\prod_{i=1}^N\mathbb{\hat{W}}_i=(A-I_n)\prod_{i=1}^N\mathbb{\hat{X}}_i$, condition \eqref{e:MC3} is satisfied.
	
	Now, one can verify that $V(x,\hat{x})=\sum_{i=1}^N(x_i-\hat x_i)^T(x_i-\hat x_i)$ is a simulation function from $\hat{\Sigma}$ to $\Sigma$
	satisfying conditions \eqref{e:SFC1} and \eqref{e:SFC2} with $\alpha(r)=r^2$, $\sigma(r)=\lambda r$, $\rho_{ext}(r)\equiv0$ $\forall r\in\R_{\ge0}$, $\varepsilon=\sum_{i=1}^{N}\gamma_i({\eta_i/2})$, where $\eta_i$ is the state set quantization parameter of abstraction $\hat\Sigma_i$ and $\gamma_i$ is the $\mathcal{K}_{\infty}$ function satisfying \eqref{eq:TI} for $\mathcal{G}_i$.  
	
	\section{Conclusion}
	In this paper, we proposed a compositional framework for the construction of finite abstractions (a.k.a. symbolic models) of interconnected discrete-time control systems. First, we used a notion of so-called storage functions in order to construct compositionally a notion of so-called simulation functions that is used to quantify the error between the output behavior of the overall interconnected concrete system and the one of its finite abstraction. 
	Furthermore, we provided an approach to construct finite abstractions together with their corresponding storage functions for a class of discrete-time control systems under some incremental passivity property.
	Finally, we demonstrated the effectiveness of the proposed results by constructing a finite abstraction of a network of linear discrete-time control systems and its corresponding simulation function in a compositional fashion and independently of the number of subsystems.

	\bibliographystyle{ieeetr}
	\bibliography{refrences}

\begin{thebibliography}{10}

\bibitem{Tazaki2008}
Y.~Tazaki and J.~I. Imura, ``Bisimilar finite abstractions of interconnected
  systems,'' in {\em 11th International Conference on Hybrid Systems:
  Computation and Control}, pp.~514--527, Berlin, Heidelberg: Springer Berlin
  Heidelberg, 2008.

\bibitem{7403879}
G.~Pola, P.~Pepe, and M.~D.~D. Benedetto, ``Symbolic models for networks of
  control systems,'' {\em IEEE Transactions on Automatic Control}, vol.~61,
  pp.~3663--3668, Nov 2016.

\bibitem{Rungger}
M.~Rungger and M.~Zamani, ``Compositional construction of approximate
  abstractions,'' in {\em Proceedings of the 18th International Conference on
  Hybrid Systems: Computation and Control}, HSCC '15, (New York, NY, USA),
  pp.~68--77, ACM, 2015.

\bibitem{7496809}
M.~Rungger and M.~Zamani, ``Compositional construction of approximate
  abstractions of interconnected control systems,'' {\em IEEE Transactions on
  Control of Network Systems}, vol.~PP, no.~99, pp.~1--1, 2016.

\bibitem{Das2004149}
K.~C. Das and P.~Kumar, ``Some new bounds on the spectral radius of graphs,''
  {\em Discrete Mathematics}, vol.~281, no.~1--3, pp.~149 -- 161, 2004.

\bibitem{7857702}
M.~Zamani and M.~Arcak, ``Compositional abstraction for networks of control
  systems: A dissipativity approach,'' {\em IEEE Transactions on Control of
  Network Systems}, vol.~PP, no.~99, pp.~1--1, 2017.

\bibitem{murat}
M.~Arcak, C.~Meissen, and A.~Packard, {\em Networks of dissipative systems}.
\newblock SpringerBriefs in Electrical and Computer Engineering, Springer
  International Publishing, 2016.

\bibitem{Girard2009566}
A.~Girard and G.~J. Pappas, ``Hierarchical control system design using
  approximate simulation,'' {\em Automatica}, vol.~45, no.~2, pp.~566 -- 571,
  2009.

\bibitem{Rock0000}
R.~T. Rockafellar and R.~Wets, {\em Variational analysis}.
\newblock Vol. 317. Springer, 2009.

\bibitem{Jiang2001857}
Z.-P. Jiang and Y.~Wang, ``Input-to-state stability for discrete-time nonlinear
  systems,'' {\em Automatica}, vol.~37, no.~6, pp.~857 -- 869, 2001.

\bibitem{Tabu}
P.~Tabuada, {\em Verification and Control of Hybrid Systems}.
\newblock New York, NY,USA: Springer, 2009.

\bibitem{ruffer}
D.~N. Tran, B.~S. R\"uffer, and C.~M. Kellett, ``Incremental stability
  properties for discrete-time systems,'' in {\em 2016 IEEE 55th Conference on
  Decision and Control (CDC)}, pp.~477--482, Dec 2016.

\bibitem{Godsil}
C.~Godsil and G.~Royle, {\em Algebraic Graph Theory}.
\newblock New York: Springer, 2001.

\end{thebibliography}

\end{document}